\documentclass[a4paper,USenglish,cleveref]{lipics-v2019}

\usepackage{amsfonts}
\usepackage{amsmath}
\usepackage{amssymb}
\usepackage{bussproofs}

\usepackage{mathtools}

\usepackage{color}

\usepackage[all]{xy}
\usepackage{bussproofs}
\usepackage{bm}

\newcommand{\eqdef}{\coloneqq}
\newcommand{\ie}{\emph{i.e.}}
\newcommand{\eg}{\emph{e.g.}}

\newcommand{\HAo}{\mathsf{HA}^\omega}

\newcommand{\T}{\mathsf{T}}
\newcommand{\z}{0}
\newcommand{\suc}{\mathsf{suc}}
\newcommand{\rec}{\mathsf{rec}}
\newcommand{\N}{\mathbb{N}}
\newcommand{\Two}{\mathbf{2}}
\newcommand{\Baire}{\N^\N}
\newcommand{\Cantor}{\Two^\N}

\newcommand{\ke}{\mathrm{ke}}
\newcommand{\rV}{\mathrm{V}}
\newcommand{\rM}{\mathrm{M}}
\newcommand{\rB}{\mathrm{B}}
\newcommand{\rS}{\mathrm{S}}

\newcommand{\fst}{\mathsf{pr}_1}
\newcommand{\snd}{\mathsf{pr}_2}

\newcommand{\inj}{\mathsf{inj}}
\newcommand{\inl}{\mathsf{inj}_1}
\newcommand{\inr}{\mathsf{inj}_2}
\newcommand{\pr}{\mathsf{pr}}
\newcommand{\tpair}{\mathsf{pair}}
\newcommand{\pair}[2]{\left\langle#1,\,#2\right\rangle}
\newcommand{\case}{\mathsf{case}}
\newcommand{\J}{\mathsf{J}}

\newcommand{\JN}{\J\N}

\newcommand{\triple}[3]{\left\langle#1,\,#2,\,#3\right\rangle}

\newcommand{\Pair}[2]{\left\langle#1,\, #2\right\rangle}

\newcommand{\cto}{\hspace{1pt}\mathord{\to}\hspace{1pt}}

\newcommand{\R}[1]{\mathrel{\mathrm{R}_{#1}}}

\newcommand{\maj}[1]{\mathrel{\triangleleft_{#1}}}

\newcommand{\E}[2]{\mathrel{\mathrm{R}^{#1}_{#2}}}

\newcommand{\nil}{\mathsf{nil}}

\newcommand{\Ge}{\mathrm{G}}

\newcommand{\ko}[1]{\langle #1 \rangle}

\newcommand{\ku}[1]{[#1]}

\title{A Gentzen-style monadic translation of G\"odel's System T}
\titlerunning{A Gentzen-style monadic translation} 

\author{Chuangjie Xu}{Ludwig-Maximilians-Universit\"at M\"unchen, Germany \and \url{http://cj-xu.github.io/}}{cj-xu@outlook.com}{https://orcid.org/0000-0001-6838-4221}{Alexander von Humboldt Foundation and LMUexcellent initiative}

\authorrunning{C.~Xu}

\Copyright{Chuangjie Xu}



\supplement{Agda development: \url{https://github.com/cj-xu/GentzenTrans}}

\acknowledgements{The author is grateful to Mart\'in Escard\'o, Hajime Ishihara, Nils K\"opp, Paulo Oliva, Thomas Powell, Benno van den Berg, Helmut Schwichtenberg for many fruitful discussions and the anonymous reviewers for their valuable suggestions and comments.}

\EventEditors{Zena M. Ariola}
\EventNoEds{1}
\EventLongTitle{5th International Conference on Formal Structures for Computation and Deduction (FSCD 2020)}
\EventShortTitle{FSCD 2020}
\EventAcronym{FSCD}
\EventYear{2020}
\EventDate{June 29--July 5, 2020}
\EventLocation{Paris, France}
\EventLogo{}
\SeriesVolume{167}
\ArticleNo{30}

\nolinenumbers

\begin{document}

\maketitle

\begin{abstract}
We introduce a syntactic translation of G\"odel's System~$\T$ parametrized by a weak notion of a monad, and prove a corresponding fundamental theorem of logical relation. Our translation structurally corresponds to Gentzen's negative translation of classical logic. By instantiating the monad and the logical relation, we reveal the well-known properties and structures of $\T$-definable functionals including majorizability, continuity and bar recursion. Our development has been formalized in the Agda proof assistant.
\end{abstract}

\section{Introduction}

Via a syntactic translation of G\"odel's System~$\T$, Oliva and Steila~\cite{OS:bar} construct functionals of bar recursion whose terminating functional is given by a closed term in System~$\T$. The author adapts their method to compute moduli of (uniform) continuity of functions $(\N \to \N) \to \N$ that are definable in System~$\T$~\cite{xu:tcont}. Inspired by the generalizations of negative translations which replace double negation by an arbitrary nucleus~\cite{EO:peirce,ishihara:GG,vdb:kuroda}, we introduce a monadic translation of System~$\T$ into itself which unifies those in~\cite{OS:bar,xu:tcont}. This monadic translation structurally corresponds to Gentzen's negative translation.

Our translation is parametrized by a monad-like structure, which we call a nucleus, but without the restriction of satisfying the monad laws. We adopt the standard technique of logical relations to show the soundness of the translation in the sense that each term of~$\T$ is related to its translation. Because the translation is parametrized by a nucleus, we have to assume that the logical relation holds for the nucleus. Such a soundness theorem is an instance of the \emph{fundamental theorem of logical relation}~\cite{statman:logical} stating that if a logical relation holds for all constants then so does it for all terms.

Monadic translations have been widely used for assigning semantics to impure languages. Our goal is instead to reveal properties enjoyed by terms of~$\T$ and to extract witnesses of these properties. For this purpose, the nuclei we work with are not extensions of $\T$, but just simple structures given by types and terms of $\T$, so that the translation remains in~$\T$ and the extracted witnesses are terms of~$\T$. The Gentzen-style translation looks simpler than the other variants~\cite{powell:ccht,uustalu:monad}. But we demonstrate its power and elegance via its various applications including majorizability, (uniform) continuity and bar recursion. Of course these properties of $\T$-definable functionals are well known~\cite{beeson:foundations,escardo:dialogue,howard:majorizable,kohlenbach:majorization,OS:bar,powell:ccht,xu:tcont}. The main contribution of the paper, however, is in obtaining these results in a single framework simply by choosing a suitable nucleus that satisfies the logical relation for the target property.

All the results in the paper are formalized in the Agda proof assistant~\cite{agda:wiki}, except the introductory section on negative translations of predicate logic. There are some differences in the Agda development. Firstly, it works with de Bruijn indices when representing the syntax of~$\T$ to avoid handling variable names. Moreover, all the logical relations are defined between the Agda (or type-theoretic) interpretations of $\T$-types. What we have proved in Agda is that the interpretation of any $\T$-term is related to the one of its translation. In this way, we avoid dealing with the computation rules of~$\T$ because they all hold judgmentally in the Agda interpretation. The Agda development is available at the author's GitHub page to which the link is given above the introduction.

\subsection{Proof-theoretic translations}
\label{sec:neg-trans}

Recall that Gentzen's translation\footnote{Nowadays it is known as the G\"odel-Gentzen negative translation. G\"odel's translation places a double negation also in front of the clause for implication, which makes it different from Gentzen's one in affine logic~\cite{AO:affine}.} simply places a double negation in front of atomic formulas, disjunctions and existential quantifiers~\cite{TS:book}. One can replace double negation with a \emph{nucleus}, that is, an endofunction $j$ on formulas such that for any formulas $A,B$ the following statements are provable:
\[
A \to jA
\qquad
(A \to jB) \to jA \to jB
\qquad
(jA)[t/x] \leftrightarrow j(A[t/x]).
\]
Nuclei are also known as \emph{lax modalities}~\cite{aczel:modality} and \emph{strong monad}~\cite{EO:peirce}. But in this paper we adopt the terminology and definition from~\cite{vdb:kuroda} which brought the technical motivation to this work. Each nucleus determines a proof-theoretic translation of intuitionistic predicate logic $\mathrm{IQL}$ into itself, consisting of a formula translation $A \mapsto A^\Ge_j$ defined as follows
\smallskip
\[
 \begin{aligned}
 (A \to B)^\Ge_j & \eqdef A^\Ge_j \to B^\Ge_j &&\quad&
 P^\Ge_j & \eqdef j P \qquad \text{for primitive $P$} \\
 (A \wedge B)^\Ge_j & \eqdef A^\Ge_j \wedge B^\Ge_j &&&
 (A \vee B)^\Ge_j & \eqdef j (A^\Ge_j \vee B^\Ge_j) \\
 (\forall x A)^\Ge_j & \eqdef \forall x A^\Ge_j &&&
 (\exists x A)^\Ge_j & \eqdef j \exists x A^\Ge_j
 \end{aligned}
\]
and a soundness theorem stating that $\mathrm{IQL} \vdash A$ implies $\mathrm{IQL} \vdash A^\Ge_j$. Working with different nuclei, one embed a logic system into another:
\begin{itemize}
\setlength\itemsep{1pt}
\item if $jA= (A \to \bot) \to \bot$, then $\mathrm{CQL}\vdash A$ implies $\mathrm{MQL} \vdash A^\Ge_j$;
\item if $jA=(A \to R) \to R$ for some predicate variable $R$, then $\mathrm{CQL}\vdash A$ implies $\mathrm{IQL} \vdash A^\Ge_j$;
\item if $jA=A\vee \bot$, then $\mathrm{IQL}\vdash A$ implies $\mathrm{MQL} \vdash A^\Ge_j$;
\end{itemize}
where $\mathrm{CQL}$ stands for classical predicate logic and $\mathrm{MQL}$ for minimal predicate logic. These results are well-known (see \eg~\cite{ishihara:GG,vdb:kuroda}) and various instances of the translation have been applied in term extraction (see \eg~\cite{EO:peirce,kohlenbach:book})

Under the viewpoint of the proofs-as-programs correspondence, our translation of G\"odel's System~$\T$ presented in Section~\ref{sec:trans} is exactly a term/program version of the above proof-theoretic translation on minimal propositional logic.

\subsection{G\"odel's System T}
\label{sec:T}

Recall that the term language of G\"odel's System $\T$ can be given by the following grammar
\[
\begin{aligned}
& \text{Type} & \sigma,\tau & \Coloneqq \N \mid \sigma\to\tau \\
& \text{Term} & t,u  & \Coloneqq x \mid \lambda x^\sigma.t \mid t u \mid \z \mid \suc \mid \rec_\sigma
\end{aligned}
\]
where $\N$ is the base type of natural numbers and $\sigma \to \tau$ the type of functions from $\sigma$ to $\tau$. A typing judgment takes the form $\Gamma \vdash t : \tau$, where $\Gamma$ is a context (\ie\ a list of distinct typed variables $x:\sigma$), $t$ is a term and $\tau$ is a type. Here are the typing rules:
\begin{center}
\AxiomC{$\Gamma,x:\sigma \vdash x : \sigma$}
\DisplayProof
\qquad
\AxiomC{$\Gamma, x:\sigma \vdash t:\tau$}
\UnaryInfC{$\Gamma \vdash \lambda x^\sigma.t : \sigma \to \tau$}
\DisplayProof
\qquad
\AxiomC{$\Gamma \vdash t : \sigma \to \tau$}
\AxiomC{$\Gamma \vdash u : \sigma$}
\BinaryInfC{$\Gamma \vdash tu:\tau$}
\DisplayProof
\end{center}
\begin{center}
\AxiomC{$\Gamma \vdash \z : \N$}
\DisplayProof
\qquad
\AxiomC{$\Gamma \vdash \suc : \N \to \N$}
\DisplayProof
\qquad
\AxiomC{$\Gamma \vdash \rec_\sigma : \sigma \to (\N \to \sigma \to \sigma) \to \N \to \sigma$}
\DisplayProof
\end{center}
We call $\Gamma \vdash t : \tau$ a \emph{well-typed} term if it is derivable.
We may omit the context $\Gamma$ and simply write $t:\tau$ or $t^\tau$ if it is unambiguous. When mentioning terms of~$\T$ in the paper, we refer to only the well-typed ones. We often omit superscript and subscript types if they can be easily inferred, and may write:
\begin{itemize}
\item $\lambda x_1 x_2 \cdots x_n. t$ instead of $\lambda x_1. \lambda x_2. \cdots \lambda x_n. t$,
\item $f(a_1,a_2,\cdots,a_n)$ instead of $(((fa_1)a_2)\cdots)a_n$,
\item $n+1$ instead of $\suc\, n$,
\item $\tau^\sigma$ instead of $\sigma \to \tau$, and
\item $f \circ g$ instead of $\lambda x.f(gx)$.
\end{itemize}
Using the primitive recursor, we can for instance define the function $\max:\N \to \N \to \N$ that returns the greater argument as follows:
\[
\max \eqdef \rec_{\N \to \N} (\lambda n^\N.n , \, \lambda n^\N f^{\N \to \N}. \rec_\N (\suc \, n, \, \lambda m^\N g^{\N \to \N}. \suc (fm))).
\]
One can easily verify that the usual defining equations of $\max$
\[
\max(\z,n) = n
\qquad
\max(m,\z) = m
\qquad
\max(\suc \, m, \suc \, n) = \suc (\max (m , n))
\]
hold using the computation rules of $\rec$
\[
\rec_\sigma(a,f,\z)=a
\qquad
\rec_\sigma(a,f,\suc\,n) = f(n,\rec_\sigma(a,f,n))
\]
where $a:\sigma$ and $f:\N \to \sigma \to \sigma$. For the ease of understanding, we will use defining equations rather than $\T$-terms involving $\rec$ in the paper.



\section{A monadic translation of System T}
\label{sec:trans}

Our syntactic translation of System~$\T$ is parametrized by a nucleus, that is, a monad-like structure without the restriction of satisfying the monad laws.

\begin{definition}[nuclei]
A \emph{nucleus} relative to~$\T$ is a triple $(\J\N,\eta,\kappa)$ consisting of a type~$\J\N$ and two terms
\[
\eta : \N \to \J\N
\qquad
\kappa : (\N \to \JN) \to \JN \to \JN
\]
of System~$\T$.
\end{definition}

Note that a nucleus is \emph{not} an extension of~$\T$, but instead a simple structure given by a type and two terms of~$\T$. Therefore, our translation of any term of~$\T$ remains in~$\T$ rather than some monadic metatheory such as in~\cite{powell:ccht,uustalu:monad}. The simplest example is the identity nucleus where $\J\N$ is just $\N$ and $\eta,\kappa$ are the identity functions of suitable types. More examples are available in Section~\ref{sec:app}. Though the first component of a nucleus is just a type, we denote it as $\J\N$ because in the generalized notion of a nucleus discussed in Section~\ref{sec:last} it will be a map $\J$ on all the types of~$\T$.

We are now ready to construct a syntactic translation of $\T$ into itself:
\begin{definition}[$\J$-translation]
\label{def:trans}
Given a nucleus $(\J\N,\eta,\kappa)$, we assign to each type $\rho$ of~$\T$ a type $\rho^\J$ as follows:
\[
\begin{aligned}
\N^\J & \eqdef \J\N \\
(\sigma \to \tau)^\J & \eqdef \sigma^\J \to \tau^\J.
\end{aligned}
\]
Each term $\Gamma \vdash t : \rho$ is translated to a term $\Gamma^\J \vdash t^\J : \rho^\J$, where $\Gamma^\J$ is a new context assigning each $x:\sigma \in \Gamma$ to a fresh variable $x^\J : \sigma^\J$, and $t^\J$ is translated inductively as follows:
\[
\begin{aligned}
(x)^\J & \eqdef x^\J & \quad & &
\z^\J & \eqdef \eta \hspace{1pt} \z \\
(\lambda x.t)^\J & \eqdef \lambda x^\J.t^\J & & &
\suc^\J & \eqdef \kappa (\eta \circ \suc) \\
(tu)^\J & \eqdef t^\J u^\J & & &
(\rec_\sigma)^\J & \eqdef \lambda x^{\sigma^\J} f^{\J\N \to \sigma^\J \to \sigma^\J} . \hspace{1pt} \ke_\sigma \left( \rec_{\sigma^\J} \left( x , f \circ \eta \right) \right)
\end{aligned}
\]
where $\ke_\sigma : (\N \to \sigma^\J) \to \J\N \to \sigma^\J$ is an extension of~$\kappa$ defined inductively on~$\sigma$:
\[
\begin{aligned}
\ke_\N & \eqdef \kappa \\
\ke_{\sigma \to \tau} & \eqdef \lambda g^{\N \to \sigma^\J \to \tau^\J} a^{\J\N} x^{\sigma^\J}. \hspace{1pt} \ke_\tau \left( \lambda n^\N. g(n,x),a \right).
\end{aligned}
\]
We often write $\J$ to denote the nucleus $(\JN,\eta,\kappa)$ and call the above the $\J$-\emph{translation} of $\T$.
\end{definition}

Thanks to the inductive translation of function types into function types, the translation of the simply-typed-$\lambda$-calculus fragment of~$\T$ is straightforward. There is no need of introducing a nonstandard, monadic notion of function application which plays an essential role in the other monadic translations~\cite{powell:ccht,uustalu:monad} as discussed in Section~\ref{sec:last}.

The more interesting part is the translation of the constants. Viewing $\eta$ as a \emph{unit} operator and $\kappa$ as a \emph{bind} operator in a monad~$\J$ on~$\N$ may reveal some intuition behind the translation of $\z$ and $\suc$: It is natural to expect ${\underline{n}}^\J = \eta \underline{n}$ for each numeral $\underline{n} \eqdef \suc^n(\z)$. This is indeed the case if the monad laws are satisfied, because $\kappa(\eta \circ -) : (\N \to \N) \to \J\N \to \J\N$ which is used to translate $\suc$ recovers exactly the `functoriality' of~$\J$. It is also natural to expect $(\rec_\sigma)^\J$ to preserve the computation rules, \ie
\[
(\rec_\sigma)^\J(x,f,\z^\J) = x
\qquad
(\rec_\sigma)^\J(x,f,(\suc\, n)^\J) = f(n^\J, (\rec_\sigma)^\J(x,f,n^\J)).
\]
A promising candidate of such $(\rec_\sigma)^\J(x,f) : \J\N \to \sigma^\J$ is $\rec_{\sigma^\J}( x , f \circ \eta) : \N \to \sigma^\J$. Hence, we extend $\kappa$ to $\ke_\sigma : (\N \to \sigma^\J) \to \J\N \to \sigma^\J$ to complete the translation of $\rec_\sigma$.

We adopt the standard technique of logical relations to show that the above translation is sound in the sense that each term of~$\T$ is related to its translation\footnote{We owe the idea of proving a unified theorem of logical relation to Thomas Powell.}. Because the translation is parametrized by a nucleus, we have to assume that the logical relation holds for the nucleus.

\begin{theorem}[Fundamental Theorem of Logical Relation]
\label{thm:LR}
Let $(\J\N,\eta,\kappa)$ be a nucleus. Given a binary relation ${\R{\N}} \subseteq \N \times \J\N$ between terms of $\T$, we extend it to ${\R{\rho}} \subseteq \rho \times \rho^\J$ for arbitrary type $\rho$ of $\T$ by defining
\[
f \R{\sigma\to\tau} g \eqdef \forall x^{\sigma} a^{\sigma^\J} \left( x \R{\sigma} a \to fx \R{\tau} ga \right).
\]
If ${\R{\N}}$ satisfies
\[
\label{eq:LR}\tag{$\dagger$}
\forall n^\N \left( n \R{\N} \eta n \right)
\qquad
\text{and}
\qquad
\forall f^{\N\to\N} g^{\N\to\JN} \left( \forall n^\N \left(fn \R{\N} gn\right) \to f \R{\N\to\N} \kappa g \right)
\]
then $t \R\rho t^\J$ for any closed term $t:\rho$ of $\T$.
\end{theorem}
\begin{proof}
We prove a more general statement that
\begin{quote}
for any term $\Gamma \vdash t : \rho$ of $\T$, if $\Gamma \R{} \Gamma^\J$ then $t \R{\rho} t^\J$
\end{quote}
where $(x_1:\sigma_1, \ldots, x_n:\sigma_n) \R{} (x_1^\J:\sigma_1^\J, \ldots, x_n^\J:\sigma_n^\J)$ stands for $x_1 \R{\sigma_1} x_1^\J \wedge \ldots \wedge x_n \R{\sigma_n} x_n^\J$, by structural induction over~$t$.
\begin{itemize}
\item $t=x$. By the assumption $\Gamma \R{} \Gamma^\J$.
\item $t=\lambda x.u$. Assume $\Gamma \R{} \Gamma^\J$ and $x \R{\sigma} x^\J$. We have $u \R{} u^\J$ by induction hypothesis.
\item $t=uv$. By induction hypothesis we have $u \R{\sigma \to \tau} u^\J$ and $v \R{\sigma} v^\J$. Then, by the definition of $\R{\sigma \to \tau}$, we have $uv \R{\tau} u^\J v^\J$.
\item $t=\z$. By the assumption (\ref{eq:LR}) of $\eta$.
\item $t=\suc$. By the assumption (\ref{eq:LR}) of $\eta$, we have $\suc(n) \R{\N} \eta(\suc(n))$ for all $n:\N$. Then by the assumption (\ref{eq:LR}) of $\kappa$, we have $\suc \R{\N} \kappa(\eta \circ \suc)$.
\item $t=\rec$. We prove $\rec \R{} \rec^\J$ with the following claims:
 \begin{enumerate}
 \item For any type $\sigma$ of $\T$, the term $\ke_\sigma$ preserves the logical relation in the following sense:
 \[
 \forall f^{\N \to \sigma} g^{\N \to \sigma^\J} \left( \forall n^\N \left (fn \R{\sigma} gn\right) \to f \R{\N\to\sigma} \ke_\sigma(g) \right).
 \]
 \textsf{Proof.}
 By induction on $\sigma$.
 \qed
 \item For any $x^\sigma$ and $y^{\sigma^\J}$ with $x \R{} y$, and any $f^{\N \to \sigma \to \sigma}$ and $g^{\JN \to \sigma^\J \to \sigma^\J}$ with $f \R{} g$,
 \[
 \forall n^\N \left( \rec_{\sigma}(x,f,n) \R{\sigma} \rec_{\sigma^\J}(y,g\circ\eta,n) \right).
 \]
 \textsf{Proof.}
 By induction on $n$.
 \qed
 \end{enumerate}
 We get a proof of $\rec \R{} \rec^\J$ simply by applying (1) to (2). \qedhere
\end{itemize}
\end{proof}

\begin{remark}
The above proof can be carried out in the intuitionistic Heyting arithmetic in finite types $\HAo$~\cite{TvD:constructivism}, with the theorem formulated as
\begin{quote}
if $\HAo$ proves (\ref{eq:LR}) then, for each closed term $t$ of $\T$, $\HAo$ proves $t \R{} t^\J$.
\end{quote}
So are all the results in Section~\ref{sec:app}. Hence, the verification system here can be~$\HAo$. We leave it unspecified in the theorem for several reasons. Firstly, we hope to study other properties whose verification may require a stronger system as in~\cite{powell:ccht}. Moreover, what we have proved in the Agda formalization is a version of the theorem for the Agda embedding of System $\T$, namely that the Agda interpretations of any $\T$-term and its translation are related. But that is only an implementation choice as explained in the introduction.
\end{remark}

\section{Applications of the monadic translation}
\label{sec:app}
We now apply the above framework to reveal various properties and structures of $\T$-definable functions including majorizability, (uniform) continuity and bar recursion. Each example consists of an algorithm to construct the desired structure given by the monadic translation, and a correctness proof of the algorithm given by the fundamental theorem of logical relation.  For this, one only needs to choose a suitable nucleus that satisfies the logical relation for the target property.

\subsection{Majorizability}
\label{sec:maj}
Our first application is to recover Howard's majorizability proof of System~$\T$~\cite{howard:majorizable}. Majorizability plays an important role in models of higher-order calculi and more recently in the proof mining program~\cite{kohlenbach:book}. Howard's majorizability relation extends the usual ordering $\leq$ on natural numbers to the one~$\maj{\rho}$ on functionals of arbitrary finite type $\rho$ in the same way as in Theorem~\ref{thm:LR}. Specifically $\maj{\rho}$ is defined inductively on $\rho$ as follows:
\[
\begin{aligned}
n \maj{\N} m & \eqdef n \leq m \\
f \maj{\sigma \to \tau} g & \eqdef \forall x^{\sigma} y^{\sigma} \left( x \maj{\sigma} y \to fx \maj{\tau} gy \right).
\end{aligned}
\]
We say $t$ is \emph{majorized} by $u$ if $t \maj{} u$, and call $u$ a \emph{majorant} of $t$. Howard shows that each closed term of T is majorized by some closed term of T, which fits perfectly into our framework: Let us take $\J\N = \N$ and define $\eta : \N \to \N$ and $\kappa : (\N \to \N) \to \N \to \N$ by
\[
\eta(n) \eqdef n
\qquad\quad
\begin{aligned}
\kappa(g,0) & \eqdef g(0) \\
\kappa(g,n+1) & \eqdef \max \left( \kappa(g,n),g(n+1) \right).
\end{aligned}
\]
The $\max$ function can be defined in $\T$ using $\rec$ as shown in Section~\ref{sec:T}, and thus so is~$\kappa$. Intuitively $\kappa(g,n)$ is the maximum of the values $g0,g1,\ldots,gn$. Therefore, it satisfies the following property:

\begin{lemma}
\label{lm:maj}
For any $g : \N \to \N$, we have $gm \leq \kappa (g,n)$ whenever $m \leq n$.
\end{lemma}
\begin{proof}
By induction on $n$.
If $n=0$, we are done because $m$ has to be $0$. If $m \leq n+1$, we have two cases: (i)~If $m=n+1$, then $g(n+1) \leq \kappa(g,n+1)$ by definition. (ii)~If $m \leq n$, then $g(m) \leq \kappa(g,n) \leq \kappa(g,n+1)$ by induction hypothesis and definition.
\end{proof}

\begin{corollary}
Each closed term $t : \rho$ of $\T$ is majorized by its translation $t^\J$.
\end{corollary}
\begin{proof}
We only need to check that the two conditions (\ref{eq:LR}) are fulfilled. The first one holds because the ordering $\leq$ is reflexive. For the second, let us assume $\forall n \left( fn \leq gn \right)$ and $n \leq m$. We have $gn \leq \kappa(g,m)$ by Lemma~\ref{lm:maj}, and thus $fn \leq \kappa(g,m)$ by the transitivity of $\leq$.
\end{proof}

We draw the reader's attention to this simple example also because the nucleus defined above does \emph{not} satisfy the monad laws: Here $\eta$ is the identity function on $\N$. If the left-identity law $\kappa(f,\eta x) = fx$ holds, then $\kappa$ has to be the identity function on $\Baire$, which is not the case.

\subsection{Lifting to higher-order functionals}
\label{sec:lift}
In the previous example, we extend a relation on natural numbers to arbitrary finite types and then show that the resulting logical relation holds for all terms of $\T$. However, if one wants to prove a certain property~$P$ of functions $X \to \N$, the above syntactic method may not work directly, because the property $P$ may not be captured by the inductively defined logical relation. Our monadic translation can serve a preliminary step to solve the problem by lifting natural numbers to functions $X \to \N$ so that the desired property~$P$ becomes the base case of the logical relation.

Let $X$ be a type of~$\T$. Consider the nucleus $(\JN,\eta,\kappa)$ with $\JN = X \to \N$ and $\eta : \N \to \JN$ and $\kappa : (\N \to \JN) \to \JN \to \JN$ defined by
\[
\eta(n) \eqdef \lambda x.n
\qquad
\kappa(g,f) \eqdef \lambda x.g(fx,x).
\]
Clearly $\eta$ maps a natural number~$n$ to a constant function with value~$n$. The intuition of $\kappa(g,f):X \to \N$ is the following: Given an input~$x$, we have an index~$fx$ to get a function $g(fx)$ from the sequence~$g$. Then we apply it to the input~$x$ to get the final value.

Given $x:X$, we define a logical relation ${\E{x}{\rho}} \subseteq \rho \times \rho^\J$ inductively as in Theorem~\ref{thm:LR}:
\[
\begin{aligned}
n \E{x}{\N} f & \eqdef n = fx \\
g \E{x}{\sigma \to \tau} h & \eqdef \forall y^\sigma, z^{\sigma^\J} \left( y \E{x}{\sigma} z \to gy \E{x}{\tau} hz \right).
\end{aligned}
\]
Clearly the conditions (\ref{eq:LR}) hold; thus by Theorem~\ref{thm:LR} we have $t \E{x}{\rho} t^\J$ for any closed term~$t$ of~$\T$. In particular, for any closed term $f: X \to \N$ of $\T$, we have
\[
\forall \Omega^{X^\J} \left( x \E{x}{X} \Omega \to fx = f^\J(\Omega,x) \right).
\]
For some type $X$ of~$\T$, we may be able to construct a closed term $\Omega : X^\J$ such that $x \E{x}{X} \Omega$ for all $x:X$, by unfolding the statement $x \E{x}{X} \Omega$. For example, if $X = \N^\N$, then $x \E{x}{X} \Omega$ is unfolded to $\forall n^\N, f^{\N^\N \to \N} \left( n=fx \to xn = \Omega(f,x) \right)$; we thus define $\Omega(f,x) \eqdef x(fx)$ as $fx=n$ by assumption and then have $x \E{x}{X} \Omega$ by definition. Once we construct such a term $\Omega : X^\J$, we have $f = f^\J \Omega$ (up to pointwise equality). The term $\Omega : X^\J$ which preserves the logical relation in the sense of $x \E{x}{X} \Omega$ for all $x:X$ is known as a \emph{generic element}~\cite{cj10,escardo:dialogue}.


Given a property~$P$ of functions $X \to \N$, we define a predicate $Q_\rho \subseteq \rho^\J$ on elements of the translated type $\rho^\J$ inductively on $\rho$:
\[
\begin{aligned}
Q_\N(f) & \eqdef P(f) \\
Q_{\sigma \to \tau}(h) & \eqdef \forall z^{\sigma^\J} \left( Q_\sigma(z) \to Q_\tau(hz) \right).
\end{aligned}
\]
Note that $Q$ is just an instance of the binary relation defined in Theorem~\ref{thm:LR}. Once we prove the conditions (\ref{eq:LR}) for $Q$, \ie
\[
\forall n^\N Q_\N(\eta n)
\qquad \text{and} \qquad
\forall g^{\N \to X \to \N} \left( \forall n^\N Q_\N(g n) \to Q_{\N \to \N}(\kappa g) \right)
\]
we have $Q(t^\J)$ for any closed term~$t$ of~$\T$. If we prove also $Q_X(\Omega)$, then we have $P(f)$ for all closed terms $f:\N \to X$ of $\T$ because $Q_\N(f^\J \Omega)$ and $f=f^\J \Omega$.


All the remaining examples are about properties of $\T$-definable functions $\N^\N \to \N$ which can be proved following the above steps. We instead enrich the `lifting' nucleus to reflect the computational content of the properties so that witnesses of the properties can be obtained as terms of~$\T$ directly via the translation.

\subsection{Continuity}
\label{sec:cont}

The next applications of our monadic translation are to recover the well-known results that every $\T$-definable function $\Baire \to \N$ is pointwise continuous and its restriction to any compact subspace is uniformly continuous~\cite{beeson:foundations}.

There are various approaches to continuity: Kohlenbach~\cite{kohlenbach:majorization} extracts a term from the extensionality proof via the Dialectica interpretation, and then uses the majorant of this term to construct a modulus of uniform continuity. Coquand and Jaber~\cite{cj10} extend type theory with a new constant for a generic element, decorate the operational semantics with forcing information, and then extract continuity information of a functional by applying it to the generic element. Escard\'o~\cite{escardo:dialogue} also employs a generic element to compute continuity but in his model of dialogue trees, which is closely related to our syntactic approach as discussed in~Remark~\ref{rm:cont}. There are also various sheaf models~\cite{escardo:xu:kk,fourman:notions,vdh:moerdijk:sheaf} in which all functionals from the Baire space ~$\N^\N$ are continuous and those from the Cantor space $\Cantor$ are uniformly continuous. Powell~\cite{powell:ccht} introduces a monadic translation for some call-by-value functional languages, one of whose instantiations tackles also continuity of $\T$-definable functionals. His method corresponds to Kuroda's negative translation as discussed in Section~\ref{sec:kuroda}.

We enrich the `lifting' nucleus (Section~\ref{sec:lift}) so that moduli of (uniform) continuity are obtained directly from the translation. For the sake of convenience, we extend System $\T$ with products. Such extension can be avoided by working with sequences of types and terms as in the literature of functional interpretations such as~\cite{oliva:unifying}.

\subsubsection{Translating products}
\label{sec:prod}
We extend System~$\T$ with product type $\sigma \times \tau$ and constants
\[
\tpair : \sigma \to \tau \to \sigma \times \tau
\qquad
\fst : \sigma \times \tau \to \sigma
\qquad
\snd : \sigma \times \tau \to \tau
\]
satisfying the usual computation rules. Similarly to Gentzen's translation of conjunction, we translate product type component-wise, \ie~$(\sigma \times \tau)^\J \eqdef \sigma^\J \times \tau^\J$. Then the above constants are translated into themselves but of the translated types, \eg~$\fst^\J \eqdef \fst : \sigma^\J \times \tau^\J \to \sigma^\J$. Recall that the primitive recursor is translated using $\ke_\sigma : (\N \to \sigma^\J) \to \JN \to \sigma^\J$ which is defined inductively on~$\sigma$. So we have to add the following case
\[
\ke_{\sigma \times \tau} \eqdef \lambda g^{\N \to \sigma^\J \times \tau^\J} a^{\J\N}. \, \tpair \left( \ke_\sigma(\fst \circ g, a) , \ke_\tau(\snd \circ g, a) \right)
\]
into the definition of $\ke$ in order to complete the translation. For the fundamental theorem of logical relation, when extending a relation ${\R{\N}} \subseteq \N \times \JN$ to ${\R{\rho}} \subseteq \rho \times \rho^\J$, we add the following case for product type
\[
u \R{\sigma \times \tau} v \eqdef \left( \fst u \R{\sigma} \fst v \right) \wedge \left( \snd u \R{\tau} \snd v \right).
\]
and can easily show that the constants of product types are related to their translations. We often write $\pair{a}{b}$ instead of $\tpair(a,b)$ for the sake of readability.

\subsubsection{Pointwise continuity}
\label{sec:pcont}
Recall that a function $M: \Baire \to \N$ is a \emph{modulus of continuity} of $f: \Baire \to \N$ if
\[
\forall \alpha^{\Baire} \beta^{\Baire} \left( \alpha =_{M\alpha} \beta \to f\alpha = f\beta \right)
\]
where $\alpha =_m \beta$ stands for $\forall i\mathord{<}m \left(\alpha i = \beta i\right)$. Our goal is to find a suitable nucleus~$\J$ so that we can obtain such a functional~$M$ from the $\J$-translation of~$f$ and then verify its correctness using the fundamental theorem of logical relation.

Let $\JN = (\Baire \to \N) \times (\Baire \to \N)$. For $w:\JN$ we write $\rV_w$ to denote its first component and $\rM_w$ the second, due to the intuition that $\rM_w$ is a modulus of continuity of the value component $\rV_w$. Then we define $\eta: \N \to \JN$ by
\[
\eta(n) \eqdef \pair{\lambda \alpha. n}{\lambda \alpha. 0}
\]
and $\kappa: (\N \to \JN) \to \JN \to \JN$ by
\[
\kappa(g,w) \eqdef \pair{\lambda \alpha. \rV_{g(\rV_w (\alpha))}(\alpha)}{\lambda \alpha. \max(\rM_{g(\rV_w (\alpha))}(\alpha),\rM_w(\alpha))}.
\]
Note that the `value' components form a `lifting' nucleus in the sense of Section~\ref{sec:lift} so that natural numbers are lifted to functions $\Baire \to \N$. And the `modulus' components will allow the translation to equip a continuity structure to the values. Reasonably $\eta(n)$ equips the constantly zero function as a modulus of continuity to the constant function $\lambda \alpha.n$ since its input is never accessed. As to $\kappa(g,w)$, its value at a point~$\alpha$ has two possible moduli: one given by $g(\rV_w(\alpha))$ and the other by $w$; thus the greater one is a modulus of continuity at $\alpha$.

We work with a logical relation ${\E{\alpha}{\rho}} \subseteq \rho \times \rho^\J$ which is parametrized by $\alpha : \Baire$. Specifically, its base case ${\E{\alpha}{\N}} \subseteq \N \times \JN$ is defined by
\[
n \E{\alpha}{\N} w \eqdef n = \rV_w(\alpha) \wedge \forall \beta \left( \alpha =_{\rM_w(\alpha)} \beta \to \rV_w(\alpha) = \rV_w(\beta) \right).
\]
The first component of $n \E{\alpha}{\N} w$ states that the value of~$w$ at~$\alpha$ is~$n$, while the second explains exactly the intuition of the type $\JN$, namely that $\rM_w(\alpha)$ is a modulus of continuity of~$\rV_w$ at~$\alpha$.
We leave the proof of (\ref{eq:LR}) to the reader. By Theorem~\ref{thm:LR}, we have $t \E{\alpha}{\rho} t^\J$ for any $\alpha : \Baire$ and for any closed term $t:\rho$ of~$\T$. In particular, we have $f \E{\alpha}{\Baire \to \N} f^\J$ for every closed term $f: \Baire \to \N$ of $\T$.

The last step is to construct the generic element $\Omega: \JN \to \JN$ such that $\alpha \E{\alpha}{\Baire} \Omega$ for all $\alpha:\Baire$. Once we unfold $\alpha \E{\alpha}{\Baire} \Omega$, we can see that, for any $w:\JN$, the value of $\Omega(w)$ has to be $\lambda \alpha. \alpha(\rV_w(\alpha))$ as discussed in Section~\ref{sec:lift}. Then we also need to construct its modulus of continuity. There are two possible moduli at $\alpha$: one is $\rV_w(\alpha)+1$ because the modulus of continuity of $\lambda \alpha.\alpha n$ at $\alpha$ is $n+1$, and the other is $\rM_w(\alpha)$. We just take the greater one and then end up with the following definition:
\[
\Omega(w) \eqdef \pair{\lambda \alpha. \alpha(\rV_w(\alpha))}{\lambda \alpha. \max(\rV_w(\alpha)+1,\rM_w(\alpha))}.
\]
One may have noticed that the above is highly similar to the definition of~$\kappa$. Indeed, we have $\Omega = \kappa(\lambda n.\pair{\lambda \alpha.\alpha n}{\lambda \alpha. n+1})$.

\begin{theorem}
\label{thm:cont}
Every closed term $f: \Baire \to \N$ of~$\T$ has a modulus of continuity given by the term $\rM_{f^\J(\Omega)}$.
\end{theorem}
\begin{proof}
Because $f \E{\alpha}{\Baire \to \N} f^\J$ and $\alpha \E{\alpha}{\Baire} \Omega$, we have $f\alpha \E{\alpha}{\N} f^\J(\Omega)$ for any $\alpha: \Baire$, which implies (i)~$f=\rV_{f^\J(\Omega)}$ up to pointwise equality, and (ii)~$\rM_{f^\J(\Omega)}$ is a modulus of continuity of $\rV_{f^\J(\Omega)}$. Therefore, $\rM_{f^\J(\Omega)}$ is also a modulus of continuity of~$f$.
\end{proof}

\begin{remark}
\label{rm:cont}
The above development can be viewed as a syntactic (and simplified) version of Escard\'o's approach via dialogue trees~\cite{escardo:dialogue}. The algorithms to construct moduli of continuity in these two methods are exactly the same. On the other hand, though Powell works also with a monadic translation~\cite{powell:ccht}, his algorithm is different because he translates terms in the call-by-value manner. We will look into this in more detail in Section~\ref{sec:kuroda}.
\end{remark}

\subsubsection{Uniform continuity}
\label{sec:ucont}

The objective here is to, for each closed term $f : \Baire \to \N$ of $\T$, construct a \emph{modulus of uniform continuity} $M: \Baire \to \N$, \ie
\[
\forall \delta^{\Baire} \alpha^{\Baire} \beta^{\Baire} \left( \alpha \leq^1 \delta \wedge \beta \leq^1 \delta \wedge \alpha =_{M\delta} \beta \to f\alpha =f\beta \right)
\]
where $\alpha \leq^1 \beta$ stands for $\forall i (\alpha i \leq \beta i)$. The value $M\delta$ is called a modulus of uniform continuity of~$f$ on $\{ \alpha : \Baire \mid \alpha \leq^1 \delta \}$. The following fact of uniform continuity plays an important role in the construction:

\begin{lemma}
\label{lm:mi}
If $f: \Baire \to \N$ is uniformly continuous on $\{ \alpha : \Baire \mid \alpha \leq^1 \delta \}$ with a modulus~$m$, then it has a maximum image on $\{ \alpha : \Baire \mid \alpha \leq^1 \delta \}$.
\end{lemma}
\begin{proof}
We compute the maximum image $\Theta(m,f,\delta)$ by induction on the modulus~$m$:
\[
\begin{aligned}
\Theta(0,f,\delta) & \eqdef f\delta \\
\Theta(m+1,f,\delta) & \eqdef \Phi \left( \lambda i. \Theta \left(m , \lambda \alpha .f(i*\alpha), \delta \circ \suc \right), \delta 0 \right)
\end{aligned}
\]
where $i*\alpha$ is an infinite sequence with head $i$ and tail $\alpha$, and $\Phi : \Baire \to \N \to \N$ defined by
\[
\begin{aligned}
\Phi(\alpha,0) & \eqdef \alpha 0 \\
\Phi(\alpha,n+1) & \eqdef \max \left( \Phi(\alpha,n), \alpha(n+1) \right)
\end{aligned}
\]
\ie~$\Phi(\alpha,n)$ is the greatest $\alpha i$ for $i\leq n$ (note that $\Phi$ is the $\kappa$ of the `majorizability' nucleus introduced in Section~\ref{sec:maj}). In the base case of $\Theta$, the modulus is 0 and thus $f$ is constant with the value $f\delta$. To compute $\Theta(m+1,f,\delta)$, by induction hypothesis we have for each~$i:\N$ the maximum image $\Theta\left(m , \lambda \alpha .f(i*\alpha), \delta \circ \suc \right))$ of the function $\lambda \alpha .f(i*\alpha)$ with a modulus~$m$ on the inputs bounded by $\delta \circ \suc$. Because the inputs of~$f$ are bounded by $\delta$, the greatest $\Theta\left(m , \lambda \alpha .f(i*\alpha), \delta \circ \suc \right))$ for $i<\delta 0$ is the maximum image of~$f$, and we use $\Phi$ to find it. Note that both $\Phi$ and $\Theta$ can be defined in $\T$ using $\rec$.
\end{proof}

We are now ready to construct the nucleus. Let $\JN = (\Baire \to \N) \times (\Baire \to \N)$. The ieda is exactly the same as in the previous treatment to pointwise continuity: For any $w:\JN$, its second component $\rM_w$ is (expected to be) a modulus of uniform continuity of the first component $\rV_w$. Then we define $\eta: \N \to \JN$ by
\[
\eta(n) \eqdef \pair{\lambda \alpha. n}{\lambda \alpha. 0}
\]
and $\kappa: (\N \to \JN) \to \JN \to \JN$ by
\[
\kappa(g,w) \eqdef \pair{\lambda \alpha. \rV_{g(\rV_w(\alpha))}(\alpha)}{\lambda \delta. \max(\Phi(\lambda i. \rM_{gi}(\delta), \Theta(\rM_w(\delta),\rV_w,\delta)), \rM_w(\delta))}
\]
where $\Phi$ and $\Theta$ are defined in the proof of Lemma~\ref{lm:mi}. Specifically, we construct a modulus of uniform continuity for the value of $\kappa(g,w)$ as follows: Given~$\delta : \Baire$, we have two possible moduli given by those of $g$ and $w$ at $\delta$ and thus we choose the larger one. Actually the only complication comes from the calculation of the modulus given by~$g$. For each $i:\N$ we have a modulus $\rM_{gi}(\delta)$. In the value of $\kappa(g,w)$, we apply $g$ to $\rV_w(\alpha)$ for input~$\alpha$. Because $\rV_w$ has a maximum image computed using $\Theta$, we only need to find the greatest modulus $\rM_{gi}(\delta)$ for $i$ not greater than the maximum image of $\rV_w$, and we use~$\Phi$ for this purpose.

Given $\delta: \Baire$, the base case ${\E{\delta}{\N}} \subseteq \N \times \JN$ of the logical relation is defined by
\[
n \E{\delta}{\N} w \eqdef n = \rV_w(\delta) \wedge \forall \alpha, \beta \left( \alpha \leq^1 \delta \wedge \beta \leq^1 \delta \wedge \alpha =_{\rM_w(\delta)} \beta \to \rV_w(\alpha) = \rV_w(\beta) \right).
\]
In words, $n \E{\delta}{\N} w$ means that $n$ is the value of $w$ at $\delta$ and that $\rM_w(\delta)$ is a modulus of uniform continuity of $\rV_w$ on $\{ \alpha : \Baire \mid \alpha \leq^1 \delta \}$. It is routine to check that both $\eta$ and $\kappa$ preserve the logical relation in the sense of (\ref{eq:LR}), which we again leave to the reader. By Theorem~\ref{thm:LR} we have $t \E{\delta}{\rho} t^\J$ for any $\delta : \Baire$ and for any closed term $t:\rho$ of~$\T$. Moreover, the generic element $\Omega: \JN \to \JN$ defined by
\[
\Omega \eqdef \kappa(\lambda n. \pair{\lambda \alpha. \alpha n}{\lambda \alpha. n+1})
\]
also preserves the logical relation in the sense of $\delta \E{\delta}{\N^\N} \Omega$ for all $\delta:\Baire$.

With a proof similar to the one of Theorem~\ref{thm:cont}, we get the following result:

\begin{theorem}
Every closed term $f: \Baire \to \N$ of~$\T$ has a modulus of uniform continuity given by the term $\rM_{f^\J(\Omega)}$.
\end{theorem}

\subsection{General bar recursion}
\label{sec:bar}
To prove Schwichtenberg's theorem~\cite{schwichtenberg:brct} that the System $\T$ definable functionals are closed under a rule-like version Spector's bar recursion of type levels 0 and 1, Oliva and Steila~\cite{OS:bar} introduce a notion of general bar recursion whose termination condition is given by decidable monotone predicates on finite sequences. As the last example, we recover their construction of general-bar-recursion functionals (\cite[Definitions~3.1 \& 3.3]{OS:bar}) via an instantiation of our translation. For this, we need the following notations:
\begin{itemize}
\item We represent decidable predicates as functions $\N^* \to \mathbf{2}$, where $\N^*$ is the type of finite sequences of natural numbers and $\mathbf{2} = \{ 0,1 \}$ is the type of booleans.
\item For any $S:\N^* \to \mathbf{2}$ and $s:\N^*$, we write $S(s)$ instead of $S(s)=1$.
\item For any $s:\N^*$, we write $|s| : \N$ to denote its length and $\hat{s} : \Baire$ the extension of $s$ with infinitely many $0$'s.
\item For any $s:\N^*$ and $n:\N$, we write $s*n : \N^*$ to denote appending $n$ to $s$.
\item For any $s:\N^*$ and $\alpha : \Baire$, we write $s*\alpha : \Baire$ to denote their concatenation.
\end{itemize}
Note that the treatment of $\N^*$ and $\mathbf{2}$ is not essential. For instance, we can represent a finite sequence $s$ by a pair $\Pair{\alpha}{n}$ and consider $s$ as the prefix of the infinite sequence $\alpha$ of length $n$ as in our Agda implementation. All the above operations on sequences are definable in $\T$.
%
%
We also need the following definitions:
\begin{itemize}
\item We call $\xi : (\N^* \to \sigma) \to (\N^* \to \sigma^\N \to \sigma) \to \N^* \to \sigma$ a functional of \emph{general bar recursion} for $S:\N^* \to \mathbf{2}$ if $\mathcal{GBR}_S(\xi)$ holds where $\mathcal{GBR}_S(\xi)$ is defined by
\[
\mathcal{GBR}_S(\xi) \eqdef \forall \, G^{\N^* \to \sigma} H^{\N^* \to \sigma^\N \to \sigma} s^{\N^*} \left\{
 \begin{aligned}
 S(s) & \to \xi(G,H,s) = G(s) \\[-3pt]
 & \hspace{4pt} \wedge \\[-3pt]
 \neg S(s) & \to \xi(G,H,s) = H(s, \lambda n^\N.\xi(G,H,s*n))
 \end{aligned}
 \right\}.
\]
\item A predicate $S$ is \emph{monotone} if $S(s)$ implies $S(s * n)$ for all $s:\N^*$ and $n:\N$.
\item For $Y : \Baire \to \N$, we say $S$ \emph{secures} $Y$ if
\[
\forall s^{\N^*} \left( S(s) \to \forall \alpha^{\Baire} \, Y(s*\alpha) = Y(\hat{s}) \right).
\]
\end{itemize}

Let $Y : \N^\N \to \N$ be a closed term of $\T$. Oliva and Steila show (i)~for any $S$ securing $Y$, from a functional of general bar recursion for $S$ we can construct a functional of Spector's bar recursion for $Y$~\cite[Theorem~2.4]{OS:bar}, and (ii)~we can construct a monotone predicate~$S$ that secures $Y$ and a functional of general bar recursion for $S$~\cite[Theorem~3.4]{OS:bar}. In this way, they give a new proof of Schwichtenberg's bar recursion closure theorem with an explicit construction of Spector's bar-recursion functionals.

We firstly construct a nucleus for general bar recursion. Fix a type $\sigma$ of $\T$. Let $\JN = (\Baire \to \N) \times (\N^* \to \mathbf{2}) \times ((\N^* \to \sigma) \to (\N^* \to \sigma^\N \to \sigma) \cto \N^* \to \sigma)$. Given $w:\JN$, we write $\rV_w, \rS_w, \rB_w$ to denote its three components. The intuition is that $\rS_w$ is a monotone predicate securing $\rV_w$ and $\rB_w$ is a functional of general bar recursion for $\rS_w$. We define $\eta : \N \to \JN$ by
\[
\eta(n) \eqdef \triple{\lambda \alpha.n}{\lambda s.1}{\lambda G H .G}
\]
and $\kappa : (\N \to \JN) \to \JN \to \JN$ by
\[
\kappa(g,w) \eqdef
\triple{\lambda \alpha. \rV_{g(\rV_w \alpha)}\alpha}
       {\lambda s. \min(\rS_w(s),\rS_{g(\rV_w \hat{s})}(s))}
       {\lambda G H . \rB_w(\lambda s. \rB_{g(\rV_w \hat{s})}(G,H,s),H)}
\]
where $\min : \N \to \N \to \N$ returns the smaller argument. Lastly, we define the generic element $\Omega : \JN \to \JN$ by
\[
\Omega \eqdef \kappa (\lambda n. \triple{\lambda \alpha. \alpha n}{\lambda s. \mathrm{Le}(n,|s|)}{\Psi n})
\]
where $\mathrm{Le} : \N \to \N \to \mathbf{2}$ has value $1$ iff its first argument is strictly smaller than the second, and $\Psi n : (\N^* \to \sigma) \to (\N^* \to \sigma^\N \to \sigma) \to \N^* \to \sigma$ is a $\T$-definable functional of bar recursion for the constant function $Y=\lambda \alpha.n$~\cite[Lemma~2.1]{OS:bar}, \ie
\[
\forall \, G^{\N^* \to \sigma} H^{\N^* \to \sigma^\N \to \sigma} s^{\N^*} \left\{
 \begin{aligned}
 n < |s| & \to \Psi n(G,H,s) = G(s) \\[-3pt]
 & \hspace{4pt} \wedge \\[-3pt]
 n \geq |s| & \to \Psi n(G,H,s) = H(s, \lambda m.\Psi n(G,H,s*m))
 \end{aligned}
 \right\}.
\]
For any $n:\N$, it is clear that $\lambda s. \mathrm{Le}(n,|s|)$ is a monotone predicate that secures $\lambda \alpha.\alpha n$, and that $\Psi n$ is a functional of general bar recursion for $\lambda s. \mathrm{Le}(n,|s|)$.

\begin{theorem}
For any closed term $Y : \Baire \to \N$ of $\T$,
\begin{enumerate}
\item $\rS_{Y^\J \Omega}$ is a monotone predicate securing $Y$, and
\item $\rB_{Y^\J \Omega}$ is a functional of general bar recursion of $\rS_{Y^\J \Omega}$.
\end{enumerate}
\end{theorem}
\begin{proof}
Given $\alpha : \Baire$, we define the base case ${\E{\alpha}{\N}} \subseteq \N \times \JN$ of the logical relation by
\[
n \E{\alpha}{\N} w \eqdef n = \rV_w(\alpha) \wedge \rS_w \; \text{is monotone} \wedge \rS_w \; \text{secures} \; \rV_w \wedge \mathcal{GBR}_{\rS_w}(\rB_w).
\]
To apply the fundamental theorem of logical relation, we need to check the conditions (\ref{eq:LR}):
\begin{itemize}
\item {It is trivial to prove $n \E{\alpha}{\N} \eta n$ for all $n:\N$.
\item As to $\kappa$, given $f:\N \to \N$ and $g:\N \to \JN$ such that $fi \E{\alpha}{\N} gi$ for all $i:\N$, our goal is to prove $f \E{\alpha}{\N \to \N} \kappa g$. Let $n:\N$ and $w:\JN$ with $n \E{\alpha}{\N} w$ be given.
\begin{itemize}
\item We have $fn = \rV_{gn}(\alpha) = \rV_{g(\rV_w(\alpha))}(\alpha) = \rV_{\kappa(g,w)}(\alpha)$ as in the previous examples.
\item Because $S_w$ is monotone and so is $\rS_{gi}$ for all $i:\N$ by assumption, the predicate $\rS_{\kappa(g,w)}$ is also monotone.
\item If $\rS_{\kappa(g,w)}(s)$, then $\rS_{w}(s)$ and $\rS_{g(\rV_w \hat{s})}(s)$ by definition. Given $\alpha:\Baire$, we have $\rV_w(s*\alpha)=\rV_w(\hat{s})$ because $\rS_w$ secures $\rV_w$. Then we have
$
\rV_{\kappa(g,w)}(s*\alpha) = \rV_{g(\rV_w(s*\alpha))}(s*\alpha) = \rV_{g(\rV_w(\hat{s}))}(s*\alpha) = \rV_{g(\rV_w(\hat{s}))}(\hat{s}) = \rV_{\kappa(g,w)}(\hat{s})
$
because $\rS_{g(\rV_w\hat{s})}$ secures $\rV_{g(\rV_w\hat{s})}$. Hence $\rS_{\kappa(g,w)}$ secures $\rV_{\kappa(g,w)}$.
\item Lastly we show that $\rB_{\kappa(g,w)}$ is a functional of general bar recursion for $\rS_{\kappa(g,w)}$. Let $G: {\N^* \to \sigma}$, $H: {\N^* \to \sigma^\N \to \sigma}$ and $s: {\N^*} $ be given. (1)~If $\rS_{\kappa(g,w)}(s)$, then $\rS_{w}(s)$ and $\rS_{g(\rV_w \hat{s})}(s)$, and thus we have $\rB_{\kappa(g,w)}(G,H,s)=G(s)$. (2)~If $\neg\rS_{\kappa(g,w)}(s)$, then we have two cases to check: (2.1) If $\rS_w(s)$, then $\neg\rS_{g(\rV_w \hat{s})}(s)$ by definition. It is not hard to show $\rB_{\kappa(g,w)}(G,H,s)=H(s,\lambda n.\rB_{\kappa(g,w)}(G,H,s*n))$. (2.2) If $\neg(\rS_w(s))$, then we always have $\rB_{\kappa(g,w)}(G,H,s)=H(s,\lambda n.\rB_{\kappa(g,w)}(G,H,s*n))$ no matter if $\rS_{g(\rV_w \hat{s})}(s)$ holds or not.
\end{itemize}
Hence, we have $fn \E{\alpha}{\N} \kappa(g,w)$.}
\end{itemize}
Given a closed term $Y: \Baire \to \N$ of $\T$, for any $\alpha:\Baire$ we have $Y \E{\alpha}{\Baire \to \N} Y^\J$ by Theorem~\ref{thm:LR}. For any $n:\N$, we have $n \E{\alpha}{\N} \triple{\lambda \alpha. \alpha n}{\lambda s. \mathrm{Le}(n,|s|)}{\Psi n}$ by definition; thus, $\alpha \E{\alpha}{\N \to \N} \Omega$ holds by (\ref{eq:LR}) for $\kappa$ which we have just proved. Hence we have $Y\alpha \E{\alpha}{\N} Y^\J\Omega$ for any $\alpha:\Baire$. From this, we get (i)~$Y=\rV_{Y^\J\Omega}$ up to pointwise equality, (ii)~$\rS_{Y^\J\Omega}$ is a monotone predicate securing $\rV_{Y^\J\Omega}$ and thus also $Y$, and (iii)~$\rB_{Y^\J\Omega}$ is a functional of general bar recursion for $\rS_{Y^\J\Omega}$.
%
\end{proof}

The above development is just a restructuring of the work of Oliva and Steila~\cite{OS:bar} that fits into our framework. But there are some small differences (or simplifications):
\begin{itemize}
\item \cite{OS:bar} requires the predicate $S$ to satisfy the bar condition $\forall \alpha^{\Baire} \exists n^\N S(\bar{\alpha}n)$, where $\bar{\alpha}n : \N^*$ is the prefix of $\alpha$ of length $n$. As pointed out by Makoto Fujiwara in a personal discussion, this condition is not needed for the result. So we remove it in the above development.
\item \cite{OS:bar} assumes that the closed terms $Y:\Baire \to \N$ are of the form $\lambda \alpha.t$ where $\alpha : \Baire$ is the only free variable in $t:\N$, and treats $\alpha$ as a special constant (for the generic element). Motivated by a version of Escard\'o's Agda development of~\cite{escardo:dialogue}, we avoid such extension by using the lifting nucleus that is introduced in Section~\ref{sec:lift}.
\item The syntactic translation of $\T$ in~\cite[Definitions 3.1 \& 3.3]{OS:bar} contains only the construction of general-bar-recursion functionals while the one of monotone securing predicates is given in the proof of the main theorem~\cite[Theorem~3.4]{OS:bar}. We combine both in the translation in order to split the constructions from the correctness proof.
\end{itemize}

As pointed out by an anonymous reviewer, this section unifies the results in Section~\ref{sec:cont} in the sense that moduli of (uniform) continuity can be defined from monotone securing bars and general-bar-recursion functionals: Let $Y:\Baire \to \N$ be a closed term of $\T$. The monotone predicate $\rS_{Y^\J\Omega}$ is a bar, \ie~$\forall \alpha^{\Baire} \exists n^\N \rS_{Y^\J\Omega}(\bar{\alpha}n)$, as shown in~\cite[Theorem~3.4]{OS:bar}. This together with the fact that it secures $Y$ implies the pointwise continuity of $Y$. The witness of the fact that $\rS_{Y^\J\Omega}$ is a bar obtained via \eg~modified realizability~\cite{kohlenbach:book} is a continuity modulus of $Y$. Our translation in Section~\ref{sec:cont} is just an explicit procedure to get these witnesses that are blurred in the proof of~\cite[Theorem~3.4]{OS:bar}. On the other hand, the reviewer points out that we can construct a modulus $M : \Baire \to \N$ of uniform continuity of $Y$ by
\[
M(\delta) \eqdef \rB_{Y^\J\Omega}(G,H^\delta,\nil)
\]
where $G(s)\eqdef 0$, $H^\delta(s,f) \eqdef 1+\max\{ fn \mid n \leq \delta(|s|) \}$ and $\nil:\N^*$ is the empty sequence. The idea is that if $s$ is a prefix of $\delta$ then $\rB_{Y^\J\Omega}(G,H^\delta,s)$ is a modulus of uniform continuity of the function $\lambda \alpha.Y(s * \alpha)$ on $\{ \alpha:\Baire \mid \alpha \leq^1 \delta \}$: If $\rS_{Y^\J\Omega}(s)$, we know $\lambda \alpha.Y(s * \alpha)$ is a constant function because $\rS_{Y^\J\Omega}$ secures $Y$. Then $\rB_{Y^\J\Omega}(G,H^\delta,s) = G(s) = 0$ is a proper modulus. If $\neg \rS_{Y^\J\Omega}(s)$, then the step functional $H^\delta$ finds the greatest value of the moduli of $\lambda \alpha.Y(s * n * \alpha)$ given by $\rB_{Y^\J\Omega}(G,H^\delta,s*n)$ for $n \leq \delta(|s|)$ and then adds 1 to get a modulus of $\lambda \alpha.Y(s * \alpha)$.

\section{Generalization and variants of the monadic translation}
\label{sec:last}

We have developed a self-translation of System~$\T$ in the spirit of Gentzen's negative translation. We conclude the paper by generalizing it to translate sums and comparing it with another two monadic translations.

\subsection{Translating sums}
\label{sec:gen}

We generalize our Gentzen-style translation to translate also sums. Let us extend $\T$ with sum type $\sigma + \tau$ and the following constants
\[
\inl : \sigma \to \sigma + \tau
\qquad
\inr : \tau \to \sigma + \tau
\qquad
\case : (\sigma \to \rho) \to (\tau \to \rho) \to \sigma + \tau \to \rho.
\]
In his translation, Gentzen places a double negation in front of disjunctions (see Section~\ref{sec:neg-trans}). Following this inspiration, the sum type $\sigma + \tau$ should be translated into $\J(\sigma^\J + \tau^\J)$. But the simple notion of a nucleus given by a type and two terms does not suffice. We have to work with the following more general notion: A nucleus $(\J,\eta,\kappa)$ consists of an endofunction~$\J$ on types of (the extension of) $\T$, and terms
\[
\eta : \sigma \to \J\sigma
\qquad
\kappa : (\sigma \to \J\tau) \to \J\sigma \to \J\tau
\]
for any types $\sigma,\tau$ of (the extension of) $\T$. Then we add the following clauses to Definition~\ref{def:trans} to complete the translation:
\[
\begin{aligned}
(\sigma + \tau)^\J & \eqdef \J(\sigma^\J + \tau^\J) & \quad & &
{\inj_i}^\J & \eqdef \eta \circ \inj_i \qquad \text{for } i \in \{1,2\} \\
\ke_{\sigma + \tau} & \eqdef \kappa & & &
\case^\J & \eqdef \lambda f g. \ke(\case(f,g)).
\end{aligned}
\]
Generalizing the fundamental theorem of logical relation to cover sums is remained as a future task.
Both natural numbers and sums, as instances of inductive types, are translated in a very similar way. Another future task is to generalize the translation to cover all (strictly positive) inductive and coinductive types.

\subsection{The Kolmogorov-style monadic translation}
\label{sec:kolmogorov}


Barthe and Uustalu's call-by-name continuation passing style transformation~\cite{BU:CPS} corresponds to Kolmogorov's negative translation. By replacing the continuation monad with a nucleus, one obtains a Kolmogorov-style monadic translation which is studied in~\cite{uustalu:monad}. Kolmogorov places a double negation in front of every subformula. Similarly we place the nucleus in front of every subtype. Hence we have to work with the more general notion of a nucleus $(\J,\eta,\kappa)$ where $\J$ is an endofunction on types. Specifically, each type $\rho$ of $\T$ is translated to $\J\ko{\rho}$ where $\ko{\rho}$ is defined inductively as follows
\[
\begin{aligned}
\ko{\N} & \eqdef \N \\
\ko{\sigma \; \Box \; \tau} & \eqdef \J\ko{\sigma} \; \Box \; \J\ko{\tau} &
\quad \text{for } \Box \in \{ \to, \times, + \}.
\end{aligned}
\]
Each term $t:\rho$ is translated to a term $\ko{t}:\J\ko{\rho}$. In order to translate function application, we have to consider a monadic notion of application. Given $f:\J(\sigma \to \J\tau)$ and $a:\sigma$, we `apply' $f$ to $a$ as
\[
f \diamond a \eqdef \kappa(\lambda g^{\sigma \to \J\tau}.ga ,\, f).
\]
Then for any terms $t:\sigma \to \tau$ and $u:\sigma$, we define $\ko{tu} \eqdef \ko{t} \diamond \ko{u}$. The rest of the translation can be found in the appendix.

\subsection{The Kuroda-style monadic translation}
\label{sec:kuroda}

There is also a Kuroda-style monadic translation of System~$\T$ which has been studied by Powell in~\cite{powell:fi,powell:ccht}, where each type $\rho$ is translated to $\J\ku{\rho}$ with $\ku{\rho}$ defined by
\[
\begin{aligned}
\ku{\N} & \eqdef \N &&\qquad&
\ku{\sigma \times \tau} & \eqdef \ku{\sigma} \times \ku{\tau} \\
\ku{\sigma \to \tau} & \eqdef \ku{\sigma} \to \J\ku{\tau} &&&
\ku{\sigma + \tau} & \eqdef \ku{\sigma} + \ku{\tau}.
\end{aligned}
\]
Note that it actually corresponds to the variant of Kuroda's negative translation where double negations are placed also in front of conclusions of implications (see~\cite[Section~6]{FO:negative}). Here we need another notion of application for the term translation: Given $f : \J(\sigma \to \J\tau)$ and $a:\J\sigma$, we `apply' $f$ to $a$ as
\[
f \bullet a \eqdef \kappa( \lambda g^{\sigma \to \J\tau}.\kappa(g,a) ,\, f).
\]
The complete translation can be found in the appendix.

Powell makes use of the Kuroda-style translation to extract moduli of continuity~\cite[Section~5]{powell:ccht}, similarly to our development in Section~\ref{sec:cont}. However, our algorithms are different because the Kuroda-style translation is call-by-value while our Gentzen-style one is call-by-name. Consider the following example. Let $t = \lambda \alpha. \rec(\alpha 0, \lambda nm.0, 1) : \Baire \to \N$ which is a constant function. If we apply the Kuroda-style translation to $t$ with the continuity nucleus (similar to the one given in Section~\ref{sec:cont} but generalized to arbitrary types of $\T$), then we get a modulus $\lambda \alpha.1$, because in the call-by-value strategy all the inputs of $\rec$ including $\alpha 0$ are evaluated. With the Gentzen-style translation, we get $\lambda \alpha.0$ because $\alpha0$ is never evaluated in the call-by-name strategy.

\bibliography{gmt-bib}

\begin{thebibliography}{10}

\bibitem{aczel:modality}
Peter Aczel.
\newblock The {R}ussell--{P}rawitz modality.
\newblock {\em Mathematical Structures in Computer Science}, 11(4):541--554,
  2001.
\newblock \href {http://dx.doi.org/10.1017/S0960129501003309}
  {\path{doi:10.1017/S0960129501003309}}.

\bibitem{agda:wiki}
{Agda community}.
\newblock The {A}gda {W}iki.
\newblock \url{https://wiki.portal.chalmers.se/agda/pmwiki.php}.

\bibitem{AO:affine}
Rob Arthan and Paulo Oliva.
\newblock On affine logic and {\l}ukasiewicz logic, 2014.
\newblock {\tt \href{https://arxiv.org/abs/1404.0570}{arXiv:1404.0570}
  [cs.LO]}.

\bibitem{BU:CPS}
Gilles Barthe and Tarmo Uustalu.
\newblock {CPS} translating inductive and coinductive types.
\newblock In {\em 2002 ACM SIGPLAN Workshop on Partial Evaluation and
  Semantics-Based Program Manipulation (PEPM'02)}, volume~37 of {\em SIGPLAN
  Notices}, pages 131--142. ACM Press, New York, 2002.
\newblock \href {http://dx.doi.org/10.1145/503032.503043}
  {\path{doi:10.1145/503032.503043}}.

\bibitem{beeson:foundations}
Michael~J. Beeson.
\newblock {\em {F}oundations of {C}onstructive {M}athematics}.
\newblock Springer, 1985.
\newblock \href {http://dx.doi.org/10.1007/978-3-642-68952-9}
  {\path{doi:10.1007/978-3-642-68952-9}}.

\bibitem{cj10}
Thierry Coquand and Guilhem Jaber.
\newblock A note on forcing and type theory.
\newblock {\em Fundamenta Informaticae}, 100(1-4):43--52, 2010.
\newblock \href {http://dx.doi.org/10.3233/FI-2010-262}
  {\path{doi:10.3233/FI-2010-262}}.

\bibitem{escardo:dialogue}
Mart\'in~H. Escard\'o.
\newblock Continuity of {G}\"odel's system {T} functionals via effectful
  forcing.
\newblock In {\em Proceedings of the Twenty-ninth Conference on the
  Mathematical Foundations of Programming Semantics (MFPS'2013)}, volume 298 of
  {\em Electronic Notes in Theoretical Computer Science}, pages 119--141.
  Elsevier, 2013.
\newblock \href {http://dx.doi.org/10.1016/j.entcs.2013.09.010}
  {\path{doi:10.1016/j.entcs.2013.09.010}}.

\bibitem{EO:peirce}
Mart{\'i}n~H. Escard{\'o} and Paulo Oliva.
\newblock The {P}eirce translation.
\newblock {\em Annals of Pure and Applied Logic}, 163(6):681--692, 2012.
\newblock \href {http://dx.doi.org/10.1016/j.apal.2011.11.002}
  {\path{doi:10.1016/j.apal.2011.11.002}}.

\bibitem{escardo:xu:kk}
Mart{\'i}n~H. Escard{\'o} and Chuangjie Xu.
\newblock A constructive manifestation of the {K}leene--{K}reisel continuous
  functionals.
\newblock {\em Annals of Pure and Applied Logic}, 167(9):770--793, 2016.
\newblock Fourth Workshop on Formal Topology (4WFTop).
\newblock \href {http://dx.doi.org/10.1016/j.apal.2016.04.011}
  {\path{doi:10.1016/j.apal.2016.04.011}}.

\bibitem{fourman:notions}
Michael~P. Fourman.
\newblock Notions of choice sequence.
\newblock In {\em The {L}.\ {E}.\ {J}.\ {B}rouwer {C}entenary {S}ymposium},
  volume 110 of {\em Studies in Logic and the Foundations of Mathematics},
  pages 91--105. Elsevier, 1982.
\newblock \href {http://dx.doi.org/10.1016/S0049-237X(09)70125-9}
  {\path{doi:10.1016/S0049-237X(09)70125-9}}.

\bibitem{howard:majorizable}
William~A. Howard.
\newblock Hereditarily majorizable functionals of finite type.
\newblock In {\em Metamathematical investigation of intuitionistic Arithmetic
  and Analysis}, volume 344 of {\em Lecture Notes in Mathematics}, pages
  454--461. Springer, Berlin, Heidelberg, 1973.
\newblock \href {http://dx.doi.org/10.1007/BFb0066739}
  {\path{doi:10.1007/BFb0066739}}.

\bibitem{ishihara:GG}
Hajime Ishihara.
\newblock A note on the {G}\"odel-{G}entzen translation.
\newblock {\em Mathematical Logic Quarterly}, 46(1):135--137, 2000.
\newblock \href
  {http://dx.doi.org/10.1002/(SICI)1521-3870(200001)46:1<135::AID-MALQ135>3.0.CO;2-R}
  {\path{doi:10.1002/(SICI)1521-3870(200001)46:1<135::AID-MALQ135>3.0.CO;2-R}}.

\bibitem{kohlenbach:majorization}
Ulrich Kohlenbach.
\newblock Pointwise hereditary majorization and some applications.
\newblock {\em Archive for Mathematical Logic}, 31(4):227--241, 1992.
\newblock \href {http://dx.doi.org/10.1007/BF01794980}
  {\path{doi:10.1007/BF01794980}}.

\bibitem{kohlenbach:book}
Ulrich Kohlenbach.
\newblock {\em Applied Proof Theory: Proof Interpretations and their Use in
  Mathematics}.
\newblock Springer Monographs in Mathematics. Springer, 2008.
\newblock \href {http://dx.doi.org/10.1007/978-3-540-77533-1}
  {\path{doi:10.1007/978-3-540-77533-1}}.

\bibitem{oliva:unifying}
Paulo Oliva.
\newblock Unifying functional interpretations.
\newblock {\em Notre Dame Journal of Formal Logic}, 47(2):263--290, 2006.
\newblock \href {http://dx.doi.org/10.1305/ndjfl/1153858651}
  {\path{doi:10.1305/ndjfl/1153858651}}.

\bibitem{FO:negative}
Paulo Oliva and Gilda Ferreira.
\newblock On the relation between various negative translations.
\newblock In {\em Logic, Construction, Computation}, volume~3 of {\em
  Mathematical Logic Series}, pages 227--258. Ontos-Verlag, 2012.
\newblock \href {http://dx.doi.org/10.1515/9783110324921}
  {\path{doi:10.1515/9783110324921}}.

\bibitem{OS:bar}
Paulo Oliva and Silvia Steila.
\newblock A direct proof of {S}chwichtenberg's bar recursion closure theorem.
\newblock {\em The Journal of Symbolic Logic}, 83(1):70--83, 2018.
\newblock \href {http://dx.doi.org/10.1017/jsl.2017.33}
  {\path{doi:10.1017/jsl.2017.33}}.

\bibitem{powell:fi}
Thomas Powell.
\newblock A functional interpretation with state.
\newblock In {\em Proceedings of the Thirty third Annual IEEE Symposium on
  Logic in Computer Science (LICS 2018)}, pages 839--848. IEEE Computer Society
  Press, July 2018.
\newblock \href {http://dx.doi.org/10.1145/3209108.3209134}
  {\path{doi:10.1145/3209108.3209134}}.

\bibitem{powell:ccht}
Thomas Powell.
\newblock A unifying framework for continuity and complexity in higher types,
  2019.
\newblock {\tt \href{https://arxiv.org/abs/1906.10719}{arXiv:1906.10719}
  [cs.LO]}.

\bibitem{schwichtenberg:brct}
Helmut Schwichtenberg.
\newblock On bar recursion of types 0 and 1.
\newblock {\em The Journal of Symbolic Logic}, 44(3):325--329, 1979.
\newblock \href {http://dx.doi.org/10.2307/2273126}
  {\path{doi:10.2307/2273126}}.

\bibitem{statman:logical}
Richard Statman.
\newblock Logical relations and the typed lambda-calculus.
\newblock {\em Information and Control}, 65(2/3):85--97, 1985.
\newblock \href {http://dx.doi.org/10.1016/S0019-9958(85)80001-2}
  {\path{doi:10.1016/S0019-9958(85)80001-2}}.

\bibitem{TS:book}
Anne~Sjerp Troelstra and Helmut Schwichtenberg.
\newblock {\em Basic Proof Theory}, volume~43 of {\em Cambridge tracts in
  theoretical computer science}.
\newblock Cambridge University Press, 2nd edition, 2000.

\bibitem{TvD:constructivism}
Anne~Sjerp Troelstra and Dirk van Dalen.
\newblock {\em Constructivism in mathematics, Vol.\ II}, volume 123 of {\em
  Studies in Logic and the Foundations of Mathematics}.
\newblock North-Holland Publishing Co., Amsterdam, 1988.

\bibitem{uustalu:monad}
Tarmo Uustalu.
\newblock Monad translating inductive and coinductive types.
\newblock In {\em Types for Proofs and Programs (TYPES~2002)}, volume 2646 of
  {\em Lecture Notes in Computer Science}, pages 299--315. Springer, 2002.
\newblock \href {http://dx.doi.org/10.1007/3-540-39185-1_17}
  {\path{doi:10.1007/3-540-39185-1_17}}.

\bibitem{vdb:kuroda}
Benno van~den Berg.
\newblock A {K}uroda-style j-translation.
\newblock {\em Archive for Mathematical Logic}, 58(5-6):627--634, 2019.
\newblock \href {http://dx.doi.org/10.1007/s00153-018-0656-x}
  {\path{doi:10.1007/s00153-018-0656-x}}.

\bibitem{vdh:moerdijk:sheaf}
Gerrit van~der Hoeven and Ieke Moerdijk.
\newblock Sheaf models for choice sequences.
\newblock {\em Annals of Pure and Applied Logic}, 27(1):63--107, 1984.
\newblock \href {http://dx.doi.org/10.1016/0168-0072(84)90035-6}
  {\path{doi:10.1016/0168-0072(84)90035-6}}.

\bibitem{xu:tcont}
Chuangjie Xu.
\newblock A syntactic approach to continuity of {T}-definable functionals.
\newblock {\em Logical Methods in Computer Science}, 16(1):22:1--22:11, 2020.
\newblock URL: \url{https://lmcs.episciences.org/6130}.

\end{thebibliography}

\appendix

\section{The Kolmogorov- and Kuroda-style monadic translations of T}
We present the Kolmogorov- and Kuroda-style monadic translations of System~$\T$ extended with products and sums, parametrized by a nucleus $(\J,\eta,\kappa)$, where $\J$ is an endofunction on types of $\T$ and $\eta : \sigma \to \J\sigma$ and $\kappa : (\sigma \to \J\tau) \to \J\sigma \to \J\tau$ are terms of $\T$. Recall the following notions of monadic application that will be needed in the translations:
\begin{itemize}
\item Given $f:\J(\sigma \to \J\tau)$ and $a:\sigma$, we define 
\[
f \diamond a \eqdef \kappa(\lambda g^{\sigma \to \J\tau}.ga ,\, f).
\]
\item Given $f:\J(\sigma \to \J\tau)$ and $a:\J\sigma$, we define
\[
f \bullet a \eqdef \kappa( \lambda g^{\sigma \to \J\tau}.\kappa(g,a) ,\, f).
\]
\end{itemize}

\begin{definition}[Kolmogorov-style monadic translation]
We assign to each type $\rho$ a type $\J\ko{\rho}$ where $\ko{\rho}$ is defined as follows
\[
\begin{aligned}
\ko{\N} & \eqdef \N \\
\ko{\sigma \; \Box \; \tau} & \eqdef \J\ko{\sigma} \; \Box \; \J\ko{\tau} &
\quad \text{for } \Box \in \{ \to, \times, + \}.
\end{aligned}
\]
Each term $\Gamma \vdash t:\rho$ is translated to a term $\ko{\Gamma} \vdash \ko{t}:\J\ko{\rho}$, where $\ko{\Gamma}$ is a new context assigning each $x:\sigma \in \Gamma$ to a fresh variable $\hat{x}:\J\ko{\sigma}$, and $\ko{t}$ is defined inductively as follows:
\[
\begin{aligned}
\ko{x} & \eqdef \hat{x} & \qquad & &
\ko{0} & \eqdef \eta(0)& \\
\ko{\lambda x.t} & \eqdef \eta(\lambda \hat{x}. \ko{t}) & \qquad & &
\ko{\suc} & \eqdef \eta(\kappa(\eta \circ \suc)) \\
\ko{tu} & \eqdef \ko{t} \diamond \ko{u} & \qquad & &
\ko{\rec} & \eqdef \eta(\kappa(\lambda a. \eta(\kappa(\lambda f. \eta(\kappa( \rec^\diamond(a,f) )))))) \\
\ko{\tpair} & \eqdef \eta(\lambda a. \eta(\lambda b. \eta (\tpair(a,b)))) & \qquad & &
\ko{\inj_i} & \eqdef \eta(\eta \circ \inj_i) \\
\ko{\pr_i} & \eqdef \eta(\kappa(\pr_i)) & \qquad & &
\ko{\case} & \eqdef \eta(\kappa(\lambda f. \eta(\kappa(\lambda g. \eta(\kappa(\case(f,g)))))))
\end{aligned}
\]
where $\rec^\diamond : \sigma \to (\JN \to \J(\J\sigma \to \J\sigma)) \to \N \to \J\sigma$ is defined by
\[
\begin{aligned}
\rec^\diamond(a,f,0) & \eqdef \eta(a) \\
\rec^\diamond(a,f,n+1) & \eqdef f(\eta n) \diamond \rec^\diamond(a,f,n).
\end{aligned}
\]
\end{definition}

\begin{definition}[Kuroda-style monadic translation]
We assign to each type $\rho$ a type $\J\ku{\rho}$ where $\ku{\rho}$ is defined as follows
\[
\begin{aligned}
\ku{\N} & \eqdef \N &&\qquad&
\ku{\sigma \times \tau} & \eqdef \ku{\sigma} \times \ku{\tau} \\
\ku{\sigma \to \tau} & \eqdef \ku{\sigma} \to \J\ku{\tau} &&&
\ku{\sigma + \tau} & \eqdef \ku{\sigma} + \ku{\tau}.
\end{aligned}
\]
Each term $\Gamma \vdash t:\rho$ is translated to a term $\ku{\Gamma} \vdash \ku{t}:\J\ku{\rho}$, where $\ku{\Gamma}$ is a new context assigning each $x:\sigma \in \Gamma$ to a fresh variable $\bar{x}:\ku{\sigma}$, and $\ku{t}$ is defined inductively as follows:
\[
\begin{aligned}
\ku{x} & \eqdef \eta(\bar{x}) & \qquad & &
\ku{0} & \eqdef \eta(0) \\
\ku{\lambda x.t} & \eqdef \eta(\lambda \bar{x}. \ku{t}) & \qquad & &
\ku{\suc} & \eqdef \eta(\eta \circ \suc) \\
\ku{tu} & \eqdef \ku{t} \bullet \ku{u} & \qquad & &
\ku{\rec} & \eqdef \eta(\lambda a. \eta(\lambda f. \eta (\rec^\bullet(a,f)))) \\
\ku{\tpair} & \eqdef \eta(\lambda \alpha. \eta(\lambda b. \eta (\tpair(a,b)))) & \qquad & &
\ku{\inj_i} & \eqdef \eta(\eta \circ \inj_i) \\
\ku{\pr_i} & \eqdef \eta(\eta \circ \pr_i) & \qquad & &
\ku{\case} & \eqdef \eta(\lambda f. \eta(\lambda g. \eta (\case(f,g))))
\end{aligned}
\]
where $\rec^\bullet : \sigma \to (\N \to \J(\sigma \to \J\sigma)) \to \N \to \J\sigma$ is defined by
\[
\begin{aligned}
\rec^\bullet(a,f,0) & \eqdef \eta(a) \\
\rec^\bullet(a,f,n+1) & \eqdef fn \bullet \rec^\bullet(a,f,n).
\end{aligned}
\]
\end{definition}

\end{document}